\documentclass{llncs}

\usepackage{graphicx}
\usepackage{latexsym}
\usepackage{amsmath}
\usepackage{amssymb}

\newcommand{\cancel}[1]{}

\def\G{\ensuremath{\mathcal{G}}}

\def\X{\ensuremath{\mathcal{X}}}

\def\Nat{\ensuremath{\mathbb{N}}}

\begin{document}

\title{Towards Network Games\\with Social Preferences}

\author{Petr Kuznetsov \and Stefan Schmid}

\institute{
  TU Berlin / Deutsche Telekom Laboratories, D-10587 Berlin,
  Germany
}

\date{}
\maketitle

\begin{abstract}
Many distributed systems can be modeled as \emph{network games}: a
collection of \emph{selfish} players that communicate in order to
maximize their individual utilities. The performance of such games
can be evaluated through the costs of the system \emph{equilibria}:
the system states in which no player can increase her utility by
unilaterally changing her behavior.
However, assuming that all players are selfish and in particular
that all players have the same utility function may not always be
appropriate. Hence, several extensions to incorporate also
altruistic and malicious behavior in addition to selfishness have
been proposed over the last years. In this paper, we seek to go one
step further and study arbitrary relationships between participants.
In particular, we introduce the notion of the \emph{social range
matrix} and explore the effects of the social range matrix on the
equilibria in a network game. In order to derive concrete results,
we propose a simplistic network creation game that captures the
effect of social relationships among players.
\end{abstract}


\section{Introduction}

Many distributed systems have an open clientele and can only be
understood when taking into account socio-economic aspects. A
classic approach to gain insights into these systems is to assume
that all players are selfish and seek to maximize their utility.
Often, the simplifying assumption is made that all players have the
same utility function. However, distributed systems are often
``socially heterogeneous'' whose participants run different clients
and protocols, some of which may be selfish while others may even
try to harm the system. Moreover, in a social network setting where
members are not anonymous, some players may be friends and dislike
certain other players. Thus, the state and evolution of the system
depends on a plethora of different utility functions. Clearly, the
more complex and heterogeneous the behavior of the different network
participants, the more difficult it becomes to understand (or even
predict) certain outcomes.

In this paper, we propose a more general approach to model the
players' utilities 
and introduce a social range matrix. This matrix specifies the
\emph{perceived} costs that are taken into account by the players
when choosing a strategy. For example, a player who maliciously
seeks to hamper the system performance has a perceived cost that
consists of the negative costs of the other players. On the other
hand, an altruistic player takes into account the costs of all other
players and strives for a socially optimal outcome. There are many
more player types in-between that care about some players but
dislike others.

In order to gain insights into the implications of different social
ranges, we consider a novel network creation game that captures the
willingness of a group of people to connect to each other. In this
game, players do not incur infinite costs if they are not connected
to some players. Rather, the utility of a player is given by the
number of other players in her $R$-neighborhood, for some parameter
$R$. For instance, in a game with $R=1$, players can only
collaborate with and benefit from their direct neighbors. Or imagine
a peer-to-peer network like Gnutella where files are searched by
flooding up to a certain radius (e.g., a time-to-live of $R=TTL=5$),
then a player is mainly interested in the data shared in her $5$-hop
neighborhood. Our motivation in using this model stems from its
simplicity which allows to exemplify and quantify the effects of
different social matrices.

\subsection{Related Work}

Over the last years, several models for distributed systems have
been proposed that go beyond purely selfish settings. For instance,
security and robustness related issues of distributed systems have
been an active field of research, and malicious faults are studied
intensively (e.g.,~\cite{byzcastro,flightpath}). To the best of our
knowledge, the first paper to study equilibria with a malicious
player is by Karakostas and Viglas~\cite{timtipp} who consider a
routing application where a single malicious player uses his flow
through the network in an effort to cause the maximum possible
damage. In order to evaluate the impact of such malicious behavior,
a coordination ratio is introduced which compares the social costs
of the worst \emph{Wardrop equilibrium} to the social costs of the
best \emph{minimax saddle-point}. In~\cite{BARalt}, implementation
problems are investigated with $k$ faulty players in the population,
but neither their number nor their identity is known. A planner's
objective then is to design an equilibrium where the non-faulty
players act according to her rules. Or in~\cite{BAR}, the authors
describe an asynchronous state machine replication protocol which
tolerates \underline{B}yzantine, \underline{A}ltruistic, and
\underline{R}ational behavior. Moscibroda et al.~\cite{bg}
discovered the existence of a so-called \emph{fear factor} in the
virus inoculation game where the presence of malicious players can
improve the social welfare under certain circumstances. This
windfall of malice has subsequently also been studied in the
interesting work by Babaioff et al.~\cite{bobby} on congestion
games.

There exists other work on game theoretic systems in which not every
participating agent acts in a rational or malicious way. In the
\emph{Stackelberg theory}~\cite{Stackl}, for instance, the model
consists of selfish players and players that are \emph{controlled by
a global leader}. The leader's goal is to devise a strategy that
induces an optimal or near optimal so-called Stackelberg
equilibrium. Researchers have recently also been interested in the
effects of altruism that co-exists with
selfishness~\cite{altruism,windfall}. For example, Meier et
al.~\cite{windfall} have shown (for a specific game played on a
social network) that friendship among players is always beneficial
compared to purely selfish environments, but that the gain from
social ties does not grow monotonically in the degree of friendship.

In contrast to the literature discussed above, we go one step
further and initiate the study of games where players can be
embedded in arbitrary social contexts and be selfish towards certain
players, be friends with some other players, and even have enemies.

In particular, we apply our framework to a novel network creation
game (for similar games, see the connection games described in
Chapter~19.2 of~\cite{mechdesignbook}). Network creation has been a
``hot topic'' for several years. The seminal work by Fabrikant et
al.~\cite{Fabrikant03} in 2003 seeks to shed light on the Internet's
architecture as built by economic agents, e.g., by Internet
providers or \emph{autonomous systems}. Recent subsequent work on
network creation in various settings includes
\cite{uraweilts,Anshelevich04,rough1,Corbo05,demainelg}. Moscibroda
et al.~\cite{lg} considered network creation games for peer-to-peer
systems. The game proposed in our paper here can be motivated by
peer-to-peer systems as well. However, in contrast to~\cite{lg}
where peers incur an infinite cost if they are not all connected to
each other, we believe that our model is more appropriate for
unstructured peer-to-peer systems.

The notion of interpersonal influence matrix, similar to our social
range matrices, is used in sociology for understanding the dynamics
of interpersonal agreement in a group of individuals (see, e.g.,
\cite{social-influence-friedkin}).

\subsection{Our Contributions}

The main contribution of this paper is the introduction and initial
study of the social range matrix which allows us to describe
arbitrary social relationships between players. For instance, social
range matrices can capture classic \emph{anarchy} scenarios where
each player is selfish, \emph{monarchy} scenarios where players only
care about one network entity, or \emph{coalitions} that seek to
support players within the same coalition but act selfishly or
maliciously towards other coalitions. Despite this generality, we
are able to derive interesting properties of such social matrices.
For instance, we show that there are matrix transformations that do
not affect the equilibria points (and the convergence behavior) of a
game.

In addition, as a case study, we analyze a simplistic social network
creation game where players can decide to which other players they
want to connect. While a new connection comes at a certain cost, a
player can also benefit from her neighborhood. That is, we assume
that the players' utility is given by the number of other players
they are connected to up to a certain horizon, minus the cost of the
links they have to pay for. For example, this game can be motivated
by unstructured peer-to-peer systems where data is usually searched
locally (in the peers' neighborhood) and overall connectivity is not
necessarily needed. We focus on this game due to its simplicity that
allows us to study the main properties of the social matrix and
exemplify the concepts. For example, in a social context where
players can choose their neighbors, it is likely that players will
connect to those players who they are friends with. We will show
that this intuition is correct and that social relationships are
indeed often reflected in the resulting network. As another example,
we show that the social welfare of monarchic societies can be higher
than that of anarchic societies if the price of establishing a
connection is relatively low; otherwise, the welfare is lower.

Our new model and the network creation game open a large number of
research directions. We understand our work as a first step in
exploring the effect of social ranges on the performance of network
games and use this paper to report on our first insights.

\subsection{Paper Organization}

The rest of the paper is organized as follows. We describe our model
and formally introduce the social range matrix in
Section~\ref{sec:model}. Section~\ref{sec:general} presents our
first insights on the properties of a social range matrix. We then
report on our case study on social network formation
(Section~\ref{sec:netcreation}). The paper is concluded with a brief
discussion and an outlook on future research directions in
Section~\ref{sec:conclusion}.

\section{Social Range Matrices and Perceived Equilibria}\label{sec:model}

In this section, we introduce the concept of a game theory where
players are embedded in a social context; in particular, we define
the social range matrix $F$ describing for each player $i$ how much
she cares about every other player $j$.


We consider a set  $\Pi$ of $n$ players (or nodes), 
$\Pi=\{0,\ldots,n-1\}$. Let $\X_i$ be the set of possible strategies
player $i$ can pursue in a given game $\G$.
A \emph{strategy profile} $s\in \X_0\times\ldots\times\X_{n-1}$
specifies a configuration, i.e., $s$ is the vector of the strategies
of all players.

%
The cost that actually arises at a player $i$ in a given strategy
profile $s$ is described by its \emph{actual} cost function
$c_a(i,s)$. However, depending on the social context a player is
situated in, it may experience a different \emph{perceived cost}
$c_p(i,s)$: While a purely selfish player may be happy
with a certain situation, another player that cares about the actual
costs of her friends may have a higher perceived cost and may want
to change her strategy to a socially better one. (Note, however,
that the distinction between ``purely selfish players'' and players
that take into account the utility or cost of other players is
artificial: Players whose action depends on other players' utilities
can be considered ``purely selfish'' as well, and simply have a
different cost function.)

Formally,
we model the perceived costs of a given player as a
linear combination of the actual costs of all other players in the
game. The \emph{social range} of player $i$ is a vector
$f_i=(f_{i0},...,f_{i(n-1)})\in\mathbb{R}^n$. Intuitively, $f_{ij}$
quantifies how much player $i$ cares about player $j$, in both a
positive (if $f_{ij}>0$) and a negative way ($f_{ij}<0$). $f_{ij}=0$
means that $i$ does not care about $j$. The social ranges of all the
players constitute the \emph{social range matrix} $F=\{f_{ij}\}$ of
the game. We will later see (Lemma~\ref{prop:factor}) that it is
sufficient to focus on normalized matrices where $\forall i,j:
-1\leq f_{ij} \leq 1$ (rather than $f_{ij}\in \mathbb{R}$).

The perceived cost of player $i$ in a strategy profile $s$ is thus
calculated as:
\[
c_p(i,s)=\sum_j f_{ij} c_a(j,s).
\]
In other words, the perceived cost of player $i$ increases with the
aggregate costs of $i$'s \emph{friends} (players $j$ with
$f_{ij}>0$) and decreases with the aggregate costs of $i$'s
\emph{enemies} (players $j$ with $f_{ij}<0$). Note that we allow a
player $i$ to value other players' costs more than her own cost,
i.e., $f_{ii}$ can be smaller than some $|f_{ij}|$, $i\neq j$. This
captures the effect of sacrificing one's own interests for the sake
(or for the harm) of others.

Henceforth, a social matrix $F$ with all $1$'s (resp., all $-1$'s,
except for $f_{ii}$) is called \emph{altruistic} (resp.,
\emph{malicious}). Generally, a social matrix with a lot of zero or
negligibly small elements describes a system with weak social ties.
Some interesting social range matrices $F$ are:
\begin{enumerate}
\item If $F$ is the identity matrix, we are in the
realm of classic game theory where each player is selfish.

\item A completely altruistic scenario is described by a social matrix $F$ consisting of $1$s only, i.e.,
$f_{ij}=1$ ($\forall i,j$). Alternatively, we can also define an
altruistic player that considers her own costs only to a small
extent ($f_{ii}=\epsilon$, for some arbitrarily small $\epsilon>0$).

\item In a situation where $\exists k$ such that $\forall i,j$: $f_{ij}=0$ except for
$f_{ik}=1$, the players only care about a single individual. We will
refer to this situation as a \emph{monarchy scenario}. (Sometimes it
makes sense to assume that players are at least a bit
self-interested and $\forall i$: $f_{ii}=\epsilon$ for an
arbitrarily small positive $\epsilon$.)

\item If $\exists k$ such that $\forall i,j$: $f_{ij}=0$ except for
$f_{kj}=1$ (and maybe $f_{kk}=\epsilon$), there is one benevolent
player that seeks to maximize the utility of all players.

\item If $\exists k$ such that for all players $i$: $f_{ii}=1$ and otherwise $0$, and
$f_{ki}=-1$, we have a selfish scenario with one malicious player
$k$ that seeks to minimize the utility of all the players.
(Alternatively, we can also postulate that for a malicious player
$k$, $f_{kk}=1$.)

\item If $\exists j,k$ such that $\forall i$: $f_{ji}=f_{ki}$, then we
  say that players $j$ and $k$ \emph{collude}: their incentives
  to deviate from a given strategy profile are identical.
(We will show in Lemma~\ref{prop:factor} that $j$ and $k$ collude
 even if $\exists \lambda>0$:   $\forall i$, $f_{ji}=\lambda f_{ki}$.)

\end{enumerate}
There are special player types to consider, e.g.:
\begin{definition}[Ignorant and Ignored Players]\label{defn:players}
A player $i$ is called \emph{ignorant} if $f_{ij}=0~~\forall j$; the
perceived cost of an ignorant player $i$ does not depend on the
actual costs. Now suppose that $F$ contains a zero column:
$f_{ji}=0, \forall j$. In this case, no player cares about $i$'s
actual cost, and we call $i$ \emph{ignored}.
\end{definition}

In game theory, (pure) \emph{Nash equilibria} are an important
solution concept to evaluate the outcomes of games. A Nash
equilibrium is defined as a situation where no player can
unilaterally reduce her cost by choosing another strategy given the
other players' strategies. In our setting, where the happiness of a
given player depends on her perceived costs, the equilibrium concept
also needs to be expressed in terms of perceived costs. We formally
define the \emph{perceived Nash equilibrium} (PNE) as follows.
\begin{definition}[Perceived Nash Equilibrium]
A strategy profile $s$ is a \emph{perceived Nash equilibrium} if for
every $s'$ that differs from $s$ in exactly one position $i$, we
have $c_p(i,s')\geq c_p(i,s)$.
\end{definition}

In order to evaluate the system performance, we study the social
cost of an equilibrium. Note that the social cost is defined with
respect to \emph{actual} costs: the \emph{social cost} of a strategy
profile $s$ is defined
    as $\textsf{Cost}(s)=\sum_j c_a(j,s)$.
A strategy profile $s$ is a \emph{social optimum} if $\forall s'$:
$\textsf{Cost}(s')\geq \textsf{Cost}(s)$.

For a given game $\G$ and a social matrix $F$, consider the ratio
between the actual cost of the worst perceived Nash equilibrium and
the cost of the social optimum. Comparing this ratio with the price
of anarchy (the ratio computed with respect to actual Nash
equilibria), we obtain the ``effect of socialization'' that captures
the benefits or disadvantages that social relations contribute to
the outcome of the game. Below we fix a game $\G$, and give some
basic properties following immediately from the definitions.

\section{Basic Properties of Social Range Matrices}\label{sec:general}

We start our analysis by examining properties of the social range
matrix.
First, observe that $F$ is invariant to row scaling.
\begin{lemma}\label{prop:factor} Let $F$ be a social matrix, and let $\lambda>0$ be an arbitrary factor. Let
$F'$ be a social matrix obtained from $F$ by multiplying a row of
$F$ by $\lambda$. Then $s$ is a perceived Nash equilibrium
w.r.t.~$F$ if and only if $s$ is a perceived Nash equilibrium with
$F'$.
\end{lemma}
\begin{proof}
Let $i$ be the player whose row is scaled. Since player $i$'s actual
costs are not affected by multiplying $f_{i\cdot}$ by $\lambda$, the
perceived costs of all other players $j\neq i$ remain the same and
hence, they still play their equilibrium strategy under $F'$.
However, also player $i$ will stick to her strategy in $F'$:
$$c_p(i,s) = \sum_j \lambda f_{ij} c_a(j,s) \leq c_p(i,s') = \sum_j
\lambda f_{ij} c_a(j,s')$$ since we know that in $F$, $c_p(i,s) =
\sum_j f_{ij} c_a(j,s) \leq c_p(i,s') = \sum_j f_{ij} c_a(j,s')$ for
all $s'$ that differ from $s$ in $i$'s strategy.\qed\end{proof}

In particular, Lemma~\ref{prop:factor} implies that we can normalize
a social matrix $F$ by $f_{ij}'=f_{ij}/\max_{\ell,k}|f_{\ell
k}|$.\footnote{Here we assume $\max_{\ell,k}|f_{\ell k}|>0$;
    otherwise, every strategy is a perceived Nash equilibrium and the price of socialization is
    the worst possible.}
Therefore, in the following, we assume normalized matrices $F$ for
which $f_{ij}\in[-1,1]$, $\forall i,j\in\{0,\ldots,n-1\}$.
%
\begin{lemma}
\label{lemma:altruistic} If $f_{ij}=1$ $~\forall i,j$, then every
social optimum is a perceived Nash equilibrium. If $f_{ij}=-1$
$~\forall i,j$, then every social minimum is a perceived Nash
equilibrium.
\end{lemma}
\begin{proof}
The proof is simple. By the definition of a social optimum $s$,
$\sum_i c_a(i,s)$ is minimal, i.e., $\nexists s'$ with $\sum_i
c_a(i,s') < \sum_i c_a(i,s)$. Thus, $s$ is also an equilibrium if
$f_{ij}=1~~\forall i,j$, as $\nexists s'$ for a given player $j$
with $c_p(j,s')=\sum_i c_a(i,s')< c_p(j,s)=\sum_i c_a(i,s)$.

Similarly for the minimum maximizing $\sum_i c_a(i,s)$ ($\nexists
s'$ with $\sum_i c_a(i,s') > \sum_i c_a(i,s)$). Profile $s$ is also
a perceived equilibrium for $f_{ij}=-1~~\forall i,j$, as $\nexists
s'$ for a given player $j$ with $c_p(j,s')=\sum_i c_a(i,s')>
c_p(j,s)=\sum_i c_a(i,s)$. \qed\end{proof}

Note however that the opposite direction is not true: there may be
games with equilibria which are not optimal, even if all players are
altruistic, namely if the game exhibits local optima.

Another special case that allows for general statements are ignorant
and ignored players (see~Definition~\ref{defn:players}). Note that
neither ignorant nor ignored players can benefit from their
unilateral actions: their perceived cost functions do not depend on
their strategies. Moreover, no player's perceived cost depends on
the actions of an ignored player. If $s$ is a perceived equilibrium,
then any strategy $s'$ that differs from $s$ only in position $i$,
where $i$ is an ignored player, is also a perceived equilibrium. In
other words, it is sufficient to determine the set of equilibria
$\textit{PNE}'$ with respect to the strategies of non-ignored
players $\Pi'$.

Existing literature also provides interesting results on the
properties and implications of certain types of social matrices. For
instance, from the work by Babaioff et al.~\cite{bobby}---and even
earlier, from the work by Karakostas and Viglas~\cite{timtipp}!---we
know that there are games where the presence of players who draw
utility from the disutility of others, can lead to an
\emph{increase} of the social welfare; this however only holds for
certain game classes that are characterized by some form of a
generalized Braess paradox. Or from the work by Meier et
al.~\cite{windfall}, it follows that in a virus inoculation game
where the social range matrix depends on the adjacency metrics of
the social network, a society can only benefit from friendship
(positive entries in the social range matrix), although \emph{not
always} in a monotonic manner.

Thus, in specific game classes, some ``corner case'' phenomena may
be observed for certain types of social matrices. In order to focus
on the principal properties of the social range, in the following we
concentrate on our network creation game. It turns out that in games
where choosing the neighbors can be a part of a player's strategy,
there is a strong correlation between the social ties and the
resulting network topology.




\section{Case Study: Network Creation}\label{sec:netcreation}

In this section, we give a formal definition of our network creation
game  and investigate the implications of different social ranges on
the formed topologies.

\subsection{A Network Creation Game}

As a use-case for employing our game-theoretic framework, we propose
a novel simple network creation game where a node (or
\emph{player})~$i$ can decide to which other nodes~$j$ she wants to
connect in an undirected graph. Establishing a connection $\{i,j\}$
(or \emph{edge}) entails a certain cost; we will assume that
connections are undirected, and that one end has to pay for it. On
the other hand, a player benefits from positive network
externalities if it is connected to other players (possibly in a
multi-hop fashion). We assume that the gain or cost of a player
depends on the number of players in her $R$-hop neighborhood, for
some parameter $R\geq 0$. For instance, a network creation game with
$R=1$ describes a situation where players can only benefit from (or
collaborate with) their direct neighbors. As motivation for larger
radii, imagine an unstructured peer-to-peer network where searching
is done by flooding up to radius $R$, and where the number of files
that can be found increases monotonically in the number of players
reached inside this radius.

Formally, the actual cost of player $i$ is given by:
$$
c_a(i,s)= \alpha \cdot s_i - g(\sum_{j=1}^R |\Gamma^j(i,s)|)
$$
where parameter $\alpha\geq 0$ denotes the cost per connection,
$s_i$ is the number of connections player $i$ pays for, and
$|\Gamma^j(i,s)|$ specifies to how many nodes node $i$ is connected
with shortest hop-distance $j$ in a graph incurred by strategy
profile $s$.
Moreover, $g: \mathbb{N}_n \rightarrow \mathbb{R}$ is a function
that specifies the utility of being in a connected group of a given
size (here $\Nat_n=\{0,\ldots,n-1\}$). For example, $g(x)=x$ denotes
that the utility grows linearly with the number of nodes within the
given radius; a super-linear utility such as $g(x)=x^2$ may be
meaningful in situations where the networking effects grow faster,
and a sub-linear utility $g(x)=\sqrt{x}$ means that marginal utility
of additional players declines with the size.
By convention, we assume that 
$g(0)=0$.

Finally, note that multiple strategy profiles (and hence perceived
Nash equilibria) can describe the same network topology where the
links are payed by different endpoints. Henceforth, for simplicity,
we will sometimes say that a given topology \emph{constitutes} (or
\emph{is}) a social optimum or an equilibrium if the corresponding
profiles are irrelevant for the statement, are clear from the
context, or if it holds for any strategy describing this network.

Given two network topologies of the same perceived costs but where
one topology has some additional edges that need to be paid by a
given player, this player is likely to prefer the other topology.
That is, it often makes sense to assume that a player does not
completely ignore the own actual cost, that is, $\forall i: f_{ii}=
\epsilon$ for an arbitrarily small $\epsilon>0$.


%

\subsection{Social Optimum and Anarchy}
%
First we describe the properties of the general network creation
game in which players behave in a selfish manner. Social optima are
characterized in the following lemma. It turns out that cliques and
trees are the most efficient networks in our game.
\begin{lemma}
\label{lem:social} Consider the network creation game
where 
$\forall x\in\Nat_{n-1}$, $g(x+1)-g(x)>\alpha/2$. Then in the case
$R=1$, the only social optimum is the clique, and in the case $R>1$,
every social optimum is a tree of diameter at most $\min(R,n-1)$.
\end{lemma}
\begin{proof}
Let $s$ be any strategy profile. We say that an edge in $s$ is
\emph{redundant} if in the strategy profile $s'$ derived from $s$ by
dropping this edge, the $R$-neighborhood of all nodes remains the
same. Every non-redundant edge connecting a player with degree $x$
to a player with degree $y$ decreases the social cost by at least
$g(x+1)-g(x)+g(y+1)-g(y)-\alpha>0$. Naturally, every social optimum
$s$ will not have redundant edges. In the case $R=1$, the clique has
the most non-redundant edges, and thus is the only topology
resulting from the social optimum.

In the case $R>1$, suppose that the network described by $s$ is not
connected and does not contain redundant edges. Then every edge
connecting the components of the graph decreases the social cost by
a positive value. Hence, we can assume that the socially optimal
topology is connected.

Now suppose that the network has diameter $R'>R$. Consider two nodes
$i$ and $j$ such that $j$ is at distance $R'$ from $i$. Then adding
an edge connecting $i$ and $j$ increases the $R$-neighborhood of
each player by at least $1$ and thus decreases the social cost.
Therefore, the diameter of the social optimum topology is at most
$\min(R,n-1)$.

Finally, since over all connected graphs, trees have the least
number of edges and hence the cost is minimized, every social
optimum results in a tree. \qed\end{proof}

%

In a selfish setting, players are less likely to connect to each
other. Indeed, even for relatively small $\alpha$, nodes remain
isolated, resulting in a poor welfare.
\begin{lemma}
\label{lem:isolated} In the network creation game, the set of
isolated nodes is a Nash equilibrium if and only if $\forall
x\in\Nat_n$, $g(x)\leq x\alpha$.
\end{lemma}
\begin{proof}
Consider the strategy profile with no edges: $\forall j: s_j=0$. If
$\forall x\in\Nat_n$, $g(x)\leq x\alpha$, then unilaterally adding
$x$ edges may only increase the individual (actual) cost by at least
$\alpha x-g(x)$, so no node has an incentive to deviate. On the
other hand, if $\exists x\in\Nat_n$, $g(x)>x\alpha$,  then every
player has an incentive to add at least $x$ edges, and thus the
``isolated'' strategy cannot be an equilibrium. \qed\end{proof}

Lemmas~\ref{lem:social} and~\ref{lem:isolated} imply that in the
case $1<\alpha<2$, the cost of the social optimum in the
\emph{linear} network creation game (when $g(x)=x$) is
$n(n-1)(\alpha/2-1)$ for $R=1$ and $(n-1)(\alpha-2)$ for $R>1$,
while the cost of the worst Nash equilibrium is $0$, i.e.,
selfishness may bring the system to a highly suboptimal state.

Below we describe the conditions under which certain topologies,
like cliques and trees of bounded diameter, constitute Nash
equilibria of the network creation game.
\begin{lemma}
\label{lem:equilibria1} In the network creation game where $R=1$,
$\forall x\in\Nat_{\lfloor n/2 \rfloor}$,
such that $\forall y\in\Nat_{n-x}$: 
$g(2x)-g(x+y)\geq \alpha(x-y)$, every $2x$-regular graph constitutes
a Nash equilibrium.
\end{lemma}
\begin{proof}
Consider the strategy in which every player establishes $x$ outgoing
links so that the resulting topology is $2x$-regular. Unilaterally
establishing $y$ (non-redundant) links instead of $x$ (for any
$y\in\Nat_{n-x}$), a player pays the cost $\alpha y-g(x+y)\geq
\alpha x-g(2x)$, so no player has an incentive to deviate.
\qed\end{proof}

In the linear case with $R=1$ and $\alpha<1$,
Lemma~\ref{lem:equilibria1} implies that the clique is the only
regular graph that results from an equilibrium: the only $x$ that
satisfies the condition is $\lfloor n/2 \rfloor$. But in general,
the resulting network may consist of up to $\lfloor n/2x \rfloor$
disconnected cliques of $2x$ players each.


\begin{lemma}
\label{lem:equilibria2} In the network creation game with $R>1$,
where $g$ is a monotonically increasing function on $\Nat_n$ such
that $\alpha<g(n-1)$, every tree of diameter at most $\min(R,n-1)$
corresponds to a Nash equilibrium.
\end{lemma}
\begin{proof}
Consider the strategy in which every node but one maintains one edge
so that the resulting graph is a tree of diameter at least
$\min(R,n-1)$. Therefore, $n-1$ players have the cost
$\alpha-g(n-1)$ and one player has the cost $-g(n-1)$. Every extra
edge would be redundant, and dropping edges increases the cost by at
least $g(n-1)-g(n)$. Thus, no player has an incentive to deviate.
\qed\end{proof}

Having described the classic setting with selfish players, we are
ready to tackle social contexts.

%
%


\subsection{Social Equilibria}

We now turn our attention to more general matrices $F$, where player
pairs $i$ and $j$ are embedded in a social context. For simplicity,
we focus on values $f_{ij}\in\{-1,0,\epsilon,1\}$ where $f_{ij}=-1$
signifies that player $i$ does not get along well with player $j$,
$f_{ij}=0$ signifies a neutral relation, and $f_{ij}=1$ signifies
friendship. We will sometimes assume that players care at least a
little bit about their own cost, i.e., $\forall i: f_{ii}=\epsilon$
for some arbitrarily small positive $\epsilon$. (This also implies
that a player will not pay for a link which is already paid for by
some other player.) We make two additional simplifications: we have
investigated the network creation game where players can only profit
from their direct neighbors (i.e., $R=1$) in more detail, and assume
a \emph{linear} scenario where the utility of being connected to
other players grows linearly in the number of contacts, that is, the
marginal utility of connecting to an additional player is constant:
we assume that $g(x)=x$.

Clearly, in this scenario, the cost $c_p(i,s)$ (and also $c_a(i,s)$)
of a player $i$ in a strategy profile $s$ is independent of
connections that are not incident to $i$. In this case, it holds
that any social matrix $F$ has a (pure) perceived equilibrium.
\begin{lemma}\label{lemma:eq}
In the linear network creation game with $R=1$, any social range
matrix $F$ has at least one pure perceived Nash equilibrium, for any
$\alpha$.
\end{lemma}
\begin{proof}
Recall that
in the $R=1$ case, a player $i$ can only benefit from her neighbors,
that is, from a connection $\{i,j\}$ that either $i$ or the
corresponding neighbor $j$ paid for. Player $i$ will pay for the
connection to player $j$ if and only if the gain from this link is
larger than the link cost $\alpha$. We have that by establishing a
new connection from player~$i$ to player~$j$, the cost changes by
$\Delta~c_p(i)= f_{ii} \cdot \left( \alpha - 1 \right) - f_{ij}
\cdot 1$. If this cost is not larger than zero, it is worthwhile for
player~$i$ to connect; otherwise it is not. On the other hand,
player~$j$ will pay for a connection to player~$i$ iff $
\Delta~c_p(j)= f_{jj}\cdot \left( \alpha - 1 \right) - f_{ji} \cdot
1 \geq 0$. As the decision of whether to connect to a given player
or not does not depend on other connections, and as links cannot be
canceled unilaterally, the resulting equilibrium network is unique
assuming that the players will not change to a strategy of equal
cost. \qed\end{proof}

Observe that the equilibrium topology found in Lemma~\ref{lemma:eq}
is only unique if the cost inequalities $\Delta~c_p$ are strict.
Moreover, a given equilibrium topology can result from different
strategy profiles, namely if there are connections where both
players have an incentive to pay for the connection to each other.

Intuitively, we would expect that the network formed by the players
reflects the social context the players are embedded in. This can be
exemplified in several ways. The following lemma shows that for the
case of binary social matrices, there are situations where the
social matrix translates directly into an equilibrium adjacency
matrix.
\begin{lemma}\label{lemma:adjacency}
In the linear network creation game with $R=1$, assume a binary
social matrix $F$ where $\forall i,j: f_{ij}\in\{0,1\}$ and where
each player is aware of her own cost, i.e., $\forall i: f_{ii}>0$.
Then, for $1<\alpha <2$, there is an equilibrium topology that can
be described by the adjacency matrix $F'$ derived from $F$ in the
following manner: (1) $\forall i: f'_{ii}=0$ and (2) if $f_{ij}=1$
for some $i,j$, then $f'_{ij}=1$ and $f'_{ji}=1$.
\end{lemma}
\begin{proof}
The claim follows from the simple observation that for $1<\alpha<2$,
a player $i$ is willing to pay for a connection to a player $j$ if
and only if $f_{ij}=1$, as the cost difference is given by
$\Delta~c_p(i)= \alpha - f_{ii}-f_{ij}$: If $f_{ij}=0$, player~$i$
only connects if $\alpha\leq 1$, and if $f_{ij}=1$, it is worthwhile
to pay for the connection as long as $\alpha\leq 2$. Therefore, as
long as $1<\alpha<2$, one endpoint will pay for the link $\{i,j\}$
(and thus: $f'_{ij}=1$ and $f'_{ij}=1$) if $f_{ij}=1$ or $f_{ji}=1$.
Clearly, it also holds that there are no loops ($\forall i:
f'_{ii}=0$). \qed\end{proof}

Note that the condition that each player cares about her own cost is
necessary for Lemma~\ref{lemma:adjacency} to hold; otherwise, if
$f_{ii}=0$, a player could trivially connect to all players as this
does not entail any connection costs. In this case, the social
matrix still describes a valid equilibrium adjacency matrix,
however, there are many other equilibria with additional edges.

\subsection{Use Case: Anarchy vs Monarchy}

A natural question to investigate in the context of social ranges is
the relationship between a completely selfish society (in game
theory also known as an \emph{anarchy}) and a society with an
outstanding individual that unilaterally determines the cost of the
players (henceforth referred to as a \emph{monarchy}); as already
mentioned, we assume that the players always care a little bit about
their own actual costs, and hence in the monarchy, let $\forall i:
f_{ii}=\epsilon$ for some arbitrarily small $\epsilon>0$, and let
$\forall i: f_{ij}=1$ where player~$j$ is the monarch (we assume
$f_{ji}=0$ for all $i\neq j$).

Interestingly, while there are situations where a monarchy yields a
higher social welfare, the opposite is also true as there are
settings that are better for anarchies. The following result
characterizes social optima, and Nash equilibria for anarchy and
monarchy settings.
\begin{lemma}\label{lemma:optanmon}
In the linear network creation game with $R=1$, the social optimum
cost is $(\alpha/2-1) n(n-1)$ if $\alpha < 2$ and $0$ otherwise, and
the anarchy has social cost $(\alpha/2-1) n(n-1)$ if $\alpha \leq 1$
and $0$ otherwise. For the monarchy, there can be multiple
equilibria (of the same cost): for any $\alpha$, there is always an
equilibrium with cost $(\alpha-2)(n-1)$; moreover, if $\alpha\leq 1$
there is an additional equilibrium with the same cost.
\end{lemma}
\begin{proof}
We consider the social optimum, the anarchy and the monarchy in
turn.

\emph{Social optimum:}
If $\alpha\leq 2$, then Lemma~\ref{lem:social} implies that any
social optimum implies the clique, with the cost $(\alpha/2-1)
n(n-1)$.  If $\alpha>2$, then every non-redundant link increases the
social cost by $\alpha-2$ and thus the set of isolated nodes has the
minimal cost, $0$.

Observe that the social cost is given by the total number of edges
$k$ in the network: $k$ edges yield a connection cost of
$k\cdot\alpha$, and the players are connected to $2k$ other players,
thus $\textsf{Cost}(s)=k\cdot\alpha-2k$.
Using Lemma~\ref{lem:social}, for the social optimum we have:
$$
\min_s \textsf{Cost}(s) = \min_k k\cdot\alpha-2k = \begin{cases}
 (\alpha/2-1)n(n-1)  & \text{, if }\alpha\leq 2\text{ (clique)}\\
  0 & \text{, otherwise (disconnected).}
\end{cases}
$$

\emph{Anarchy:} In a purely selfish setting, a player connects to
another player if and only if $\alpha\leq 1$. By
Lemmas~\ref{lem:social} and~\ref{lem:isolated}, if $\alpha\leq 1$,
then the resulting equilibrium topology is the clique and the cost
is thus $(\alpha/2-1) n(n-1)$, and if $\alpha<1$, then the resulting
topology is the set of isolated nodes and the cost is $0$.

$$
\textsf{Cost}(\text{Nash equilibrium}) =
\begin{cases}
 (\alpha/2-1) n(n-1)  & \text{, if }\alpha\leq 1 \text{ (clique)}\\
  0 & \text{, otherwise (disconnected).}
\end{cases}
$$

\emph{Monarchy:} Let $j$ denote the monarch and let $i\neq j$ denote
any other player. Since the marginal utility of an additional
neighbor of $j$ is one while the connection cost is arbitrarily
small ($\alpha\epsilon$), a player~$i$ will always connect (i.e.,
pay for the connection) to the monarch. On the other hand, the
monarch will connect to a player if and only if $\alpha \leq 1$. The
social cost of the network equilibrium is therefore always
$(\alpha-2)(n-1)$ (up to the arbitrarily small $\epsilon$ components
in the cost), for any $\alpha$. \qed\end{proof}

Using Lemma~\ref{lemma:optanmon}, we can compare the efficiency of
the different social settings. For relatively low connection costs,
a setting with a monarch gives stronger incentives for nodes to
connect, and thus socially more preferable outcomes emerge. On the
other hand, for high connection costs, due to the selfless behavior
of the players ignoring their own connection prices, an anarchy is
preferable. As a concrete example, according to
Lemma~\ref{lemma:optanmon}, for $\alpha=3/2$, the equilibrium
network of the monarchic society is a star of utility $(n-1)/2$
while in the anarchy nobody will connect, yielding a utility of
zero. On the other hand, for $\alpha=3$, the anarchy again has zero
utility, while in the monarchy, players still connect which results
in a negative overall utility of $-(n-1)$. Thus, the following lemma
holds.
\begin{lemma}
There are situations where the social welfare of anarchy is higher
than the welfare in a monarchy, and vice versa.
\end{lemma}

\subsection{Windfall of Friendship and Price of Ill-Will}

An interesting property of our network creation game is that more
friendship relations cannot worsen an equilibrium.
\begin{lemma}\label{lemma:numofones}
Consider a social range matrix $F$ where $\forall i,j:
f_{ij}\in\{0,1\}$ and $f_{ii}=1$. Let $F'$ be a social range matrix
derived from $F$ where a non-empty set $\mathcal{N}$ of
$0$-entries in $F$ are flipped to $1$. 
Then, for any equilibrium strategy $s^F$ with respect to a social
matrix $F$, there is an equilibrium strategy $s^{F'}$ with
$\textsf{Cost}(s^{F'})\leq \textsf{Cost}(s^F)$.
\end{lemma}
\begin{proof}
We prove the claim by showing that for any equilibrium strategy
$s^F$ for $F$, there is an equilibrium strategy $s^{F'}$ for $F'$
that has at least as many connections as $s^F$. Moreover, it holds
that an equilibrium with more connections always implies a higher
social welfare.

First, recall from Lemma~\ref{lemma:eq} that an equilibrium $s^F$
always exists. Now fix such an equilibrium $s^F$ from which we will
construct the equilibrium $s^{F'}$. If $i$ and $j$ are connected in
$s^F$, then they are still connected in $s^{F'}$: as $R=1$, whether
or not a connection $\{i,j\}$ between two players $i$ and $j$ is
established depends on the actual costs $c_a(i,\cdot)$ and
$c_a(j,\cdot)$ only. If two players $i$ and $j$ are not connected in
$s^F$, they have an incentive to connect in $s^{F'}$ if $f'_{ij}=1$
and $\alpha\leq 2$. Thus, $s^{F'}$ contains a superset of the
connections in $s^F$. Now let $k$ be the number of edges in a given
profile $s$. The social cost is then given by $\textsf{Cost}(s)=
k\alpha - 2k$. For $\alpha\leq 2$, this function is monotonically
decreasing, which implies the claim. On the other hand, for $\alpha
>2$, the set of isolated nodes constitutes the only equilibrium.
\qed\end{proof}

Lemma~\ref{lemma:numofones} implies that the best equilibrium with
respect to $F'$ cannot be worse than the best equilibrium with
respect to $F$. On the other hand, it is easy to see that a similar
claim also holds for the \emph{worst} equilibrium: Consider the
equilibrium $s^{F'}$ with the fewest connections $k'$. Then, there
is an equilibrium $s^F$ with $k\leq k'$ edges: either $s^F=s^{F'}$,
or some edges can be removed. We have the following claim.
\begin{corollary}
Consider a social range matrix $F$ where $\forall i,j:
f_{ij}\in\{0,1\}$ and $f_{ii}=1$. Let $F'$ be a social range matrix
that is derived from $F$ by flipping one or several $0$ entries to
$1$. Let $s_{w}^F$ and $s_{b}^F$ be the worst and the best
equilibrium profile w.r.t.~$F$, and let $s_{w}^{F'}$ and
$s_{b}^{F'}$ be the worst and best equilibrium profile w.r.t.~$F'$
(maybe $s_{w}^F=s_{b}^F$ and/or $s_{w}^{F'}=s_{b}^{F'}$). It holds
that $\textsf{Cost}(s_w^F)\geq \textsf{Cost}(s_w^{F'})$ and
$\textsf{Cost}(s_b^F)\geq \textsf{Cost}(s_b^{F'})$.
\end{corollary}
A analogous result can be obtained for settings where players
dislike each other.
\begin{lemma}\label{lemma:numofminusones}
Consider a social range matrix $F$ where $\forall i,j:
f_{ij}\in\{-1,0\}$ and $f_{ii}=1$. Let $F'$ be a social range matrix
derived from $F$ where a non-empty set $\mathcal{N}$ of $0$-entries
in $F$ are flipped to $-1$s., where $|\mathcal{N}|\geq 1$. Then, for
any equilibrium strategy $s^F$ with respect to a social matrix $F$,
there is an equilibrium strategy $s^{F'}$ with
$\textsf{Cost}(s^F)\leq \textsf{Cost}(s^{F'})$.
\end{lemma}
\begin{proof}
First recall from the proof of Lemma~\ref{lemma:numofones} that the
social welfare increases with the total number of connections given
that $\forall i: f_{ii}=1$, and that it follows from
Lemma~\ref{lemma:eq} that an equilibrium $s^F$ always exists. Fix an
equilibrium $s^F$ to construct the equilibrium $s^{F'}$. Similarly
to the arguments used in the proof of Lemma~\ref{lemma:numofones},
if $i$ and $j$ are not connected in $s^F$, then they are
disconnected in $s^{F'}$ as well. On the other hand, if player $i$
pays for the connection to player $j$ in $s^F$, she has an incentive
to disconnect in $s^{F'}$ if $f'_{ij}=-1$ and for any non-negative
$\alpha$. Thus, $s^F$ contains a superset of the connections in
$s^{F'}$. \qed\end{proof}

\section{Conclusions and Open Questions}\label{sec:conclusion}

We understand our work as a further step in the endeavor to shed
light onto the socio-economic phenomena of today's distributed
systems which typically consist of a highly heterogeneous
population. In particular, this paper has initiated the study of
economic games with more complex forms of social relationships. We
introduced the concept of social range matrices and studied their
properties. Moreover, we have proved the intuition right (under
certain circumstances) that in our novel network creation game, the
social relationships are reflected in the network topology.

This paper reported only on a small subset of the large number of
questions opened by our model, and we believe that there remain many
exciting directions for future research. For instance,
it is interesting to study which conditions are necessary and
sufficient for counter-intuitive phenomena such as the fear factor
and the windfall of malice~\cite{bobby,bg}, or the non-monotonous
relationship between welfare and friendship in social
networks~\cite{windfall}. Another open question is the
characterization of all topologies that correspond to a Nash
equilibrium.






\end{document}